\documentclass[12pt]{article}
\usepackage{graphicx}
\usepackage{amsmath}
\usepackage{amsfonts}
\usepackage{amssymb}
\newtheorem{theorem}{Theorem}
\newtheorem{acknowledgement}[theorem]{Acknowledgement}

\newtheorem{lemma}[theorem]{Lemma}

\newenvironment{proof}[1][Proof]{\textbf{#1.} }{\ \rule{0.5em}{0.5em}}
\numberwithin{equation}{section}

\begin{document}

\title{Positive gravitational energy in arbitrary dimensions}
\author{Yvonne Choquet-Bruhat}
\maketitle

\bigskip

\emph{R\'{e}sum\'{e}}

On d\'{e}montre un th\'{e}or\`{e}me d'\'{e}nergie gravitationnelle positive en
dimension quelconque utilisant seulement des spineurs li\'{e}s au groupe
$Spin(n)$ sur une section d'espace Riemannienne $(M^{n},g)$

\emph{Abstract}

We present a streamlined, complete proof, valid in arbitrary space dimension
$n$, and using only spinors on the oriented Riemannian space $(M^{n};g),$ of
the positive energy theorem in General Relativity.

\textbf{Version fran\c{c}aise abr\'{e}g\'{e}e.}

Un espacetemps Einsteinien est une vari\'{e}t\'{e} Lorentzienne $(\mathbf{M}%
^{n+1},\mathbf{g})$ qui satisfait les \'{e}quations d'Einstein
\[
\mathbf{R}_{\alpha\beta}-\frac{1}{2}\mathbf{g}_{\alpha\beta}\mathbf{R}%
=T_{\alpha\beta};
\]
on supposera qu'il satisfait la condition d'\'{e}nergie dominante, $u_{\alpha
}T^{\alpha\beta}$ temporel pour tout vecteur temporel $u.$ Sur chaque section
spatiale $M^{n}$ la m\'{e}trique induite $g$ et la courbure extrinsique $K$
satisfont les contraintes qui s'\'{e}crivent dans un rep\`{e}re orthonorm\'{e}
d'axes $e_{i}$ tangents et $e_{0}$ orthogonal \`{a} $M^{n}$%
\begin{equation}
\mathbf{R}_{0j}\equiv{\partial}_{j}K_{h}^{h}-{D}_{h}K^{h}{}_{j}=T_{0j}%
\end{equation}%
\begin{equation}
\mathbf{S}_{00}\equiv\frac{1}{2}\{R-|K|^{2}+(\mathrm{tr}K)^{2}\}=T_{00}.
\end{equation}

On suppose que $M^{n}$ est l'union d'un compact $W$ et d'un nombre fini $N$
d'ensembles $\Omega_{I},$ appel\'{e}s bouts (ends), diffeomorphes au
compl\'{e}ment d'une boule de $R^{n}.$ On utilise une partition lisse $f_{I},
$ $f_{K}$ de l'unit\'{e} sur $M^{n}$ de supports contenus dans un $\Omega_{I}$
ou un ouvert $W_{K}$ diff\'{e}omorphe \`{a} une boule de $R^{n},$ l'union
(finie) des $W_{K}$ recouvrant $W.$ Tous ces ouverts sont munis de
coordonn\'{e}es locales $x^{i}$ et de la m\'{e}trique Euclidienne $e\equiv
\eta_{ij}dx^{i}dx^{j}.$ Un tenseur $u$ sur $M^{n}$ est une somme de
tenseurs\ $u_{I}=f_{I}u,$\ $u_{K}=f_{K}u.$ On utilise la m\'{e}trique
euclidienne pour d\'{e}finir les normes de ces tenseurs. Un espace de Banach
$C_{\beta}^{k}$ ou de Hilbert $H_{s,\delta}$ d'un tenseur $u$ sur $M^{n}$ est
d\'{e}fini \`{a} l'aide du $sup$ ou de la somme des normes des tenseurs
$u_{I}$ et $u_{K},$ des choix diff\'{e}rents de partition de l'unit\'{e}
donnent des normes \'{e}quivalentes. La vari\'{e}t\'{e} Riemannienne
$(M^{n},g)$ est dite asymptotiquement Euclidienne (A.E) si
\begin{equation}
h:=g-\underline{g}\in H_{s,\delta}\cap C_{n-2}^{1},\text{ \ }s>\frac{n}%
{2}+1,\text{ \ }\frac{n}{2}-2>\delta>-\frac{n}{2},\text{ }%
\end{equation}
o\`{u} $\underline{g}$ est une m\'{e}trique lisse identique dans chaque
$\Omega_{I}$ \`{a} la m\'{e}trique Euclidienne $e.$

La d\'{e}finition de masses gravitationelles $m_{I}$ et de moments $p_{I}$,
dits ADM, d'un espace-temps with A.E. section $(M^{n},g)$ and \ $K\in
H_{s-1,\delta+1}$ provient de la formulation Hamiltonienne des \'{e}quations
d'Einstein, on a dans $\Omega_{I}$
\begin{equation}
m_{I}:=\lim_{r\rightarrow\infty}\frac{1}{2}\int_{S_{r}^{n-1}}(\frac{\partial
g_{ij}}{\partial x^{j}}-\frac{\partial g_{jj}}{\partial x^{i}})n_{i}\mu
_{e},\text{ \ }r\equiv\{\sum_{i}(x^{i})^{2}\}^{\frac{1}{2}},
\end{equation}%
\begin{equation}
p_{I}^{h}:=\lim_{r\rightarrow\infty}\int_{S_{r}^{n-1}}P^{ih}n_{i}\mu
_{e},\text{ \ \ \ \ }P^{ih}:=K^{ih}-g^{ih}\mathrm{tr}K.
\end{equation}
On reprend l'id\'{e}e spinorielle de Witten pour d\'{e}montrer, mais en
utilisant seulement un spineur sur $M^{n}$ li\'{e} \`{a} l'alg\`{e}bre de
Clifford $\mathcal{C}l(n)$, la positivit\'{e} de la masse d'un espace temps
Einsteinien quand $R\geq0,$ donc sous la condition d\'{e}nergie dominante,
quand $M^{n}$ est une hypersurface maximale . Une formulation simple li\'{e}e
au moment $P,$ qui ne fait intervenir que la m\^{e}me sorte de spineurs sur
$M^{n},$ permet de montrer que $m_{I}\geq|p_{I}|$ donc $m\geq|p|$ sans
condition sur $R$ ni autre hypoth\`{e}se sur les sources. Les d\'{e}%
monstrations reposent sur un th\'{e}or\`{e}me d'existence pour la solution
d'une \'{e}quation de Dirac compl\'{e}t\'{e}e, elliptique, sur une
vari\'{e}t\'{e} asymptotiquement Euclidienne.

\textbf{English version}

\section{Introduction}

The most elegant and convincing proof of the positive energy theorem is by
using spinors, as did Witten\footnote{{\footnotesize For references prior to
1983 one can consult my survey on positive energy theorems for les
Houches\ 1983\ school\ reproduced in Y.Choquet-Bruhat 'General Relativity and
the Einstein equations'', Oxford University press 2009.}} in the case $n=3$
inspired by heuristic works of Deser and Grisaru originating from
supergravity. The aim of this Note is to present a streamlined, complete
proof, valid in arbitrary space dimension $n$, and using only spinors on the
oriented Riemannian space $(M^{n};g),$ without invoking spacetime spinors.

We first give the notations and the definitions we use.

\section{Definitions.}

\subsection{Asymptotically Euclidean space.}

$M^{n}$ is a smooth manifold union of a compact set $W$ and a finite number of
sets $\Omega_{I},$ diffeomorphic to the complement of a ball in $R^{n}.$ One
covers $W$ by a finite number of open sets $W_{K}$ each diffeomorphic to a
ball in $R^{n}.$ We denote by $x^{i}$ local coordinates for a domain
$\Omega_{I}$ or $W_{K}.$ We set $r:=\sum(x^{i})^{2}\}^{\frac{1}{2}}$ and take
$r_{0}>0$ such that $\Omega_{I}:=$ $\{r>r_{0}\},$ $\Omega_{I}\cap W_{K}%
=\not 0$ if $r<2r_{0}.$ We consider a preparation of $M^{n},$ i.e a smooth
partition of unity, $f_{I},f_{K},$ $f_{K}$ with support in $W_{k},$ $f_{I}$
support in $\Omega_{I}$ and $f_{I}=1$ for $r>2r_{0}.$ The Riemannian metric
$g$ is continuous and uniformly bounded above and below in each $\Omega_{I},$
$W_{K}$ by constant positive definite quadratic forms. A tensor field $u$ on
$M^{n}$ is written as $u\equiv\sum_{I}u_{I}+\sum_{K}u_{K}$ with\ \ $u_{I}%
:=f_{I}u,$ \ \ $u_{K}:=f_{K}u.$ Norms on spaces of tensor fields are defined
through their components in the $\Omega_{I},$ $W_{K}$, each endowed with the
Euclidean metric $e:=\eta_{ij}dx^{i}dx^{j}\equiv\sum(dx^{i})^{2},$ with
pointwise norm $|.|$ and volume element $\mu_{e}.$ We use the Banach and
Hilbert spaces $C_{\beta}^{k}$ and $H_{s,\delta}$ with norms
\begin{equation}
||u||_{C_{\beta}^{k}}\equiv\sup_{I,K}\{\sup_{\Omega_{I}}(r^{\beta
+k}|\underline{D}^{k}u_{I}|),\sup_{W_{K}}|\underline{D}^{k}u_{K}|\},\text{
\ \ }\underline{D}^{k}:=\frac{\partial^{k}}{\partial x^{i_{1}}...\partial
x^{i_{k}}}.
\end{equation}%
\begin{equation}
||u||_{H_{s,\delta}}^{2}:=\sum_{I=1,...N}\int_{\Omega_{I}}\sum_{0\leq k\leq
s}r^{2(k+\delta)}|\underline{D}^{k}u_{I}|^{2}\mu_{e}+\sum_{K=1,...N^{\prime}%
}\int_{W_{K}}\sum_{0\leq k\leq s}|\underline{D}^{k}u_{K}|^{2}\mu_{e}%
\end{equation}

Different preparations of $M^{n}$ give equivalent norms. A Riemannian manifold
($M^{n},g)$ is called asymptotically Euclidean (A.E) if
\begin{equation}
h_{I}:=f_{I}(g-e)\in H_{s,\delta}\cap C_{n-2}^{1},\text{ \ }f_{K}g\in
H_{s},\text{ \ \ }s>\frac{n}{2}+1,\text{ \ }\frac{n}{2}-2>\delta>-\frac{n}%
{2}.\text{ }%
\end{equation}

It can be proved (using the fact that $H_{s,\delta}$ is an algebra if
$s>\frac{n}{2},\delta>-\frac{n}{2})$ that an A.E $(M^{n},g)$ admits in each
end $\Omega_{I}$ an orthonormal coframe
\begin{equation}
\theta^{j}:=a_{i}^{j}dx^{i}\text{ \ \ \ \ \ }a_{i}^{j}=\delta_{i}^{j}+\frac
{1}{2}\lambda_{i}^{j},\text{ \ \ }\lambda_{i}^{j}\in H_{s,\delta}\cap
C_{n-2}^{1}.
\end{equation}
In the following, components in the coordinates $x^{i}$ are underlined. In
$\Omega_{I}$ it holds that
\[
\underline{g_{ij}}\equiv\sum_{h}a_{i}^{h}a_{j}^{h}\equiv\eta_{ij}%
+\underline{h_{ij}},\text{ \ \ }\eta_{ij}:=\delta_{i}^{j},\text{ \ }%
\]%
\begin{equation}
\underline{h_{ij}}\equiv\frac{1}{2}(\lambda_{i}^{j}+\lambda_{j}^{i})+\frac
{1}{4}\sum_{h}\lambda_{j}^{h}\lambda_{i}^{h},\text{ \ \ }\lambda_{j}%
^{h}\lambda_{i}^{h}\in H_{s,2\delta+\frac{n}{2}}\cap C_{2n-4}^{1}.
\end{equation}
The rotation coefficients $c_{ij}^{h}$ of the coframe $\theta^{h}$ are, with
($b_{i}^{j})$ the matrix inverse of ($a_{j}^{i})$ and $\partial_{i}$ the Pfaff
derivative with respect to $\theta^{i}$, $d\theta^{h}\equiv\frac{1}{2}%
c_{ij}^{h}\theta^{i}\wedge\theta^{j},$
\begin{equation}
\text{ \ }c_{ij}^{h}\equiv b_{j}^{k}\partial_{i}\lambda_{k}^{h}-b_{i}%
^{k}\partial_{j}\lambda_{i}^{h}\equiv\frac{1}{2}(\partial_{i}\lambda_{j}%
^{h}-\partial_{j}\lambda_{i}^{h})+\chi_{ij}^{h},\text{ \ \ }\chi_{ij}^{h}\in
H_{s-1,2\delta+1+\frac{n}{2}}%
\end{equation}
We choose the coframe such that
\[
\partial_{i}(\lambda_{i}^{j}-\lambda_{j}^{i})\in H_{s-1,2\delta+1+\frac{n}{2}}%
\]
The components $\omega_{i,jh}\equiv\frac{1}{2}(-c_{jh}^{i}+c_{ij}^{h}%
-c_{ih}^{j})$ of the Riemannian connection $\omega$ in the coframe $\theta
^{i}$ are then computed to be
\begin{equation}
\omega_{i,hj}\equiv\frac{1}{2}\underline{\partial}_{h}\underline{h_{ij}}%
-\frac{1}{2}\underline{\partial}_{j}\underline{h_{ih}}+\zeta_{i,hj},\text{
\ }\zeta_{i,hj}\in H_{s-1,2\delta+1+\frac{n}{2}}.
\end{equation}

\subsection{Global mass $m$ and linear momentum $p$}

We say that $(M^{n},g,K)$ is A.E. if $(M^{n},g)$ is A.E. and $K\in$
$H_{s-1,\delta+1}\cap C_{n-1}^{0}.$ The mass $m$ and linear momentum $p$
associated to an end $\Omega$ define a spacetime vector $\mathbf{E}$ with
components
\begin{equation}
E^{0}:=m:=\lim_{r\rightarrow\infty}\frac{1}{2}\int_{S_{r}^{n-1}}%
(\underline{\partial}_{j}\underline{h_{ij}}-\underline{\partial}_{i}%
\underline{h_{jj}})n_{i}\mu_{\bar{g}},\text{ \ \ }E^{h}:=p^{h}:=\lim
_{r\rightarrow\infty}\int_{S_{r}^{n-1}}P^{ih}n_{i}\mu_{\bar{g}}.
\end{equation}
The uniform bound of $n_{i}$ and the equivalence of $\mu_{\bar{g}}$ with
$r^{n-1}\mu_{S_{1}^{n-1}}$ show that the limits exist.

We always assume that the constraints 0.1 and 0.2 are satisfied on $M^{n}$ and
that $T$ obeys the dominant energy condition.

\section{Spinor fields and Dirac operator.}

The gamma matrices associated with an orthonormal coframe $\theta^{i}$ of $g $
at $x\in M^{n}$ are linear endomorphisms of a complex vector space $S$ of
dimension $p:=2^{[n/2]}$ which satisfy the identities
\begin{equation}
\gamma_{i}\gamma_{j}+\gamma_{j}\gamma_{i}=2\eta_{ij}I_{p},\text{
}i,j=1,...n,\text{ \ }I_{p}\text{ identity matrix}%
\end{equation}
The $\gamma_{i}$ are chosen hermitian, i.e. $\gamma_{i}=\tilde{\gamma}_{i}$ ,
as is possible for an $O(n)$ group.

The spinor group $Spin(n),$ double covering of $SO(n)$, can be realized by the
group of invertible linear maps $\Lambda$ of $S$ which satisfy, with
$O:=(O_{i}^{j})$ a $n\times n$ orthogonal matrix
\begin{equation}
\Lambda\gamma^{i}\Lambda^{-1}=O_{j}^{i}\gamma^{j}\text{ \ \ \ and \ \ }%
\det\Lambda=1.
\end{equation}

In a subset $\Omega_{I}$ or $W_{K}$ with a given field $\rho_{0}$ of
orthonormal frames a spinor field $\psi$ is represented by a mapping
$(x^{i})\mapsto$ $\underline{\psi}(x^{i})\in S.$ Under an $O\in SO(n)$ change
of frame, $\rho=O\rho_{0}$ $\ $the spinor $\psi$\ becomes represented by
$\underline{\psi}^{\prime}=\Lambda\underline{\psi}$ where some choice has been
made for the correspondence between $\Lambda$\ and $O.$ This can be made
consistently on $M^{n}$ if it admits a spin structure; that is, a homomorphism
\ of a $Spin(n)$ principal bundle $P_{Spin_{n}}$ onto the principal bundle of
oriented orthonormal frames. It is a topological property of $M^{n},$ the
vanishing of its second Stiefel-Whitney class, always true for an orientable
$M^{3}.$ A spinor field on $M^{n}$ is then a section of a vector bundle
$\Psi_{Spin(n)}$ \ associated with $P_{Spin(n)},$ with base $M^{n}$ and
typical fiber $S$. To a space of spinors corresponds a space of cospinors,
replacing $S$ by the adjoint (complex dual) vector space $\tilde{S}$ and the
change of representation by $\underline{\phi^{\prime}}=\underline{\phi}%
\Lambda^{-1}$. Using dual frames $e_{A}$ of $S$ and $\theta^{A}$ of $\tilde
{S}$ we have $\psi\equiv\psi^{A}e_{A},$ \ $\phi=\theta^{A}\phi_{A},$
\ \ \ $A=1,...p,$ we denote the duality relation by
\[
\phi\psi\equiv(\phi,\psi):=\phi_{A}\psi^{A},\text{ \ a frame independent
scalar.}%
\]
By 3.2 $\tilde{\psi}$ represented by $\tilde{\psi}_{A}:=(\psi^{A})^{\ast}$ is
a cospinor if $\psi$ is a spinor, $|\psi|^{2}\equiv\tilde{\psi}\psi$ is
positive definite.

\emph{A spin connection} $\sigma$ on ($M^{n},g)$ is deduced from an $O(n)$
connection $\omega$ by the isomorphism between the Lie algebras of $O(n)$ and
$Spin(n)$ obtained by differentiation of 3.2, it is represented in each domain
of the preparation by
\begin{equation}
\sigma_{i}\equiv\frac{1}{4}\gamma^{h}\gamma^{k}\omega_{i,hk},\text{
\ \ }i=1,...n.
\end{equation}

The covariant derivative of a spinor $\psi,$ resp. cospinor $\phi,$ is a
covariant vector spinor, resp. cospinor, with components in the frames
$\theta^{i}\otimes e_{A},$ resp. $\theta^{i}\otimes\theta^{A},$%
\[
(D_{i}\psi)^{A}\equiv\partial_{i}\psi^{A}+(\sigma_{i}\psi)^{_{A}},\ (D_{i}%
\phi)_{A}\equiv\partial_{i}\phi_{A}-(\phi\sigma_{i})_{A}.
\]
The hermiticity of $\gamma_{i}$ and 3.1 show that $\tilde{\sigma}_{i}%
=-\sigma_{i},$ hence $\widetilde{D_{i}\psi}\equiv D_{i}\tilde{\psi}.$

The Riemannian connection together with the spin connection define a first
order derivation operator mapping tensor-spinor-cospinor fields into tensor-
spinor- cospinor fields with one more covariant index. \emph{The gamma
matrices are the components of a vector-spinor-cospinor which has covariant
derivative zero}.

\emph{The spin curvature} $\rho$ is a 2-tensor- spinor -cospinor, image by the
mapping of Lie algebras of the curvature tensor of $g.$ The Ricci identity for
spinors reads
\begin{equation}
D_{i}D_{j}\psi-D_{j}D_{i}\psi\equiv\rho_{ij}\psi\text{ \ \ with \ \ }\rho
_{ij}:=\frac{1}{4}R_{ij,hk}\gamma^{h}\gamma^{k}.
\end{equation}

The Dirac operator on sections of the vector bundle $\Psi(n)$ reads locally
\begin{equation}
\mathcal{D}\psi\equiv\gamma^{i}D_{i}\psi\equiv\gamma^{i}(\partial_{i}%
\psi+\frac{1}{4}\omega_{i,hk}\gamma^{h}\gamma^{k}\psi),\text{ \ hence
\ }\widetilde{\mathcal{D}\psi}\equiv D_{i}\tilde{\psi}\gamma^{i}.
\end{equation}
The algebraic Bianchi identity together with 3.1 and 3.4 lead to the
formula\footnote{{\footnotesize See for instance A. Lichnerowicz \ Bull. Soc.
Math. France 92, 1964 p. 11-100}}
\begin{equation}
\mathcal{D}^{2}\psi\equiv\eta^{ij}D_{i}D_{j}\psi+\frac{1}{2}\gamma^{i}%
\gamma^{j}\rho_{ij}\psi\equiv\eta^{ij}D_{i}D_{j}\psi-\frac{1}{4}R\psi.
\end{equation}
The Dirac operator is a first order linear operator with principal symbol
($\eta^{ij}\xi_{i}\xi_{j})^{\frac{p}{2}},$ hence elliptic. Weighted Sobolev
spaces for spinor fields on a prepared $M^{n}$ are defined as for tensor
fields after setting $\psi=\sum(f_{I}\psi+f_{K}\psi)$ and using
representations $\underline{\psi}$. A known theorem\footnote{{\footnotesize Y.
Choquet-Bruhat and D. Christodoulou, Acta Mathematica 146, 1981.}.} gives:

\begin{theorem}
On an A.E. $(M^{n},g)$ the Dirac operator is a Fredholm operator from spinors
in $H_{s,\delta}$ to spinors in $H_{s-1,\delta+1},$ it is an isomorphism if
injective. The same is true of $\mathcal{D}\psi+f\psi$ if $f $ is a bounded
linear map from spinors in $H_{s,\delta}$ to spinors in $H_{s-1,\delta+1}.$
\end{theorem}

\section{Gravitational mass.}

We prove for arbitrary $n>2$ the fundamental fact used by Witten for $n=3.$

\begin{theorem}
Let $(M^{n},g)$ be A.E. The mass $m$ of an end $\Omega_{I}$ is equal to
\begin{equation}
\lim_{r\rightarrow\infty}\int_{S_{r}^{n-1}}\mathcal{U}_{0}^{i}n_{i}\mu
_{\bar{g}}=\frac{m}{2},\text{ \ \ }\mathcal{U}_{0}^{i}:=\mathcal{R}%
e\{\tilde{\psi}_{0}(\eta^{ij}-\gamma^{i}\gamma^{j})\sigma_{j}\psi_{0}\},
\end{equation}
with $S_{r}^{n-1}$ the submanifold of the end $\Omega_{I}$ with equation
$\{\sum(x^{i})^{2}\}^{\frac{1}{2}}=r,$ $n_{i}$ its unit normal, $\mu_{\bar{g}}
$ the volume element induced by $g$ and $\psi_{0}$ a spinor constant in
$\Omega_{I}$ (i.e. $\frac{\partial\psi_{0}}{\partial x^{i}}=0)$ and $|\psi_{0}|=1.$
\end{theorem}

\begin{proof}
We first remark that, using $\gamma^{i}=\tilde{\gamma}^{i}$ and $\tilde
{\sigma}_{i}=-\sigma_{i},$ one finds
\[
\mathcal{R}e(\tilde{\psi}_{0}\eta^{ij}\sigma_{j}\psi_{0})\equiv\frac{1}%
{2}\tilde{\psi}_{0}\eta^{ij}(\sigma_{j}+\tilde{\sigma}_{j})\psi_{0}\equiv0.
\]
The definition 3.3 of $\sigma_{j}$ then implies
\[
\mathcal{U}_{0}^{i}=\frac{1}{8}\sum_{j,h,k}\tilde{\psi}_{0}\omega
_{j,hk}(\gamma^{i}\gamma^{j}\gamma^{h}\gamma^{k}-\gamma^{h}\gamma^{k}%
\gamma^{j}\gamma^{i})\psi_{0}.
\]
The property 2.7 of $\omega$ on A.E $(M^{n},g)$ shows that $r^{n-1}%
(\omega_{j,hk}-\frac{1}{2}\underline{\partial}_{h}\underline{h_{jk}}+\frac
{1}{2}\underline{\partial}_{k}\underline{h_{jh}})$ tends uniformly to zero as
$r$ tends to infinity. The uniform bound of $n_{i}$ and the equivalence of
$\mu_{\bar{g}}$ with $r^{n-1}\mu_{S_{1}^{n-1}}$ show that the limit in 4.1
exists. Calculations\footnote{{\footnotesize Similar to those done by
P.\ Chrusciel in his Krakow lectures on Energy in General Relativity.}} using
3.1, 2.7 and the symmetry of $\underline{\partial}_{k}\underline{h_{hj}}$ in
$h$ and $j$ lead to
\[
\lim_{r\rightarrow\infty}\int_{S_{r}^{n-1}}\mathcal{U}_{0}^{i}n_{i}\mu
_{\bar{g}}=\frac{1}{2}\lim_{r\rightarrow\infty}\int_{S_{r}^{n-1}}\omega
_{j,ij}|\psi_{0}|^{2}n_{i}\mu_{\bar{g}}=\frac{1}{2}m\text{ \ \ if \ }|\psi
_{0}|^{2}=1.
\]
\end{proof}

To study the positivity of the mass one defines a vector $\mathcal{U}^{i}$ on
$M^{n}$ such that the integrals on $S_{r}^{n-1}$ of $\mathcal{U}^{i}$ and
$\mathcal{U}_{0}^{i}$ have the same limit when $r\rightarrow\infty.$ The
Stokes formula applied to the integral of the of divergence of $\mathcal{U}%
^{i}$ will give information on this limit. We set
\begin{equation}
\mathcal{U}^{i}:=\mathcal{R}e\{\tilde{\psi}(\eta^{ij}D_{j}\psi-\gamma
^{i}\gamma^{j}D_{j}\psi)\}.
\end{equation}

\begin{lemma}
On an A.E manifold $(M^{n},g)$ it holds that

1.
\[
D_{i}\mathcal{U}^{i}\geq0\text{ \ \ if \ }R\geq0\text{ \ \ \ and
\ \ }\mathcal{D}\psi=0.
\]

2. If $\psi=\psi_{0}+\psi_{1}$ with $\underline{\partial_{i}}\psi_{0}=0 $ in
\ $\Omega_{I}$ and $\psi_{1}\in H_{s,\delta}$ then in $\Omega_{I}$%
\begin{equation}
\lim_{r\rightarrow\infty}\int_{S_{r}^{n-1}}\mathcal{U}_{0}^{i}n_{i}\mu
_{\bar{g}}=\lim_{r\rightarrow\infty}\int_{S_{r}^{n}}\mathcal{U}^{i}n_{i}%
\mu_{\bar{g}}.
\end{equation}
\end{lemma}

\begin{proof}
1. By elementary computation, using $D_{i}\tilde{\psi}=\widetilde{D_{i}\psi} $
and the identity 3.6 one finds
\begin{equation}
D_{i}\mathcal{U}^{i}\equiv|D\psi|^{2}-|\mathcal{D}\psi|^{2}+\frac{1}{4}%
R|\psi|^{2}%
\end{equation}
Therefore $D_{i}\mathcal{U}^{i}\geq0$ if $R\geq0$\ and $\psi$\ satisfies the
equation $\mathcal{D}\psi=0.$

2.\ To study the limit of the integral on $S_{r}^{n-1}$ of $\mathcal{U}^{i}$
when $\psi=\psi_{0}+\psi_{1}$ we write
\begin{equation}
\mathcal{U}^{i}=\mathcal{U}_{0}^{i}+\frac{1}{2}\mathcal{R}e\{\tilde{\psi}%
_{0}[\gamma^{j},\gamma^{i}]D_{j}\psi_{1})+\tilde{\psi}_{1}[\gamma^{j}%
,\gamma^{i}]D_{j}\psi)\}.
\end{equation}
Hence $D_{i}\mathcal{U}^{i}=D_{i}\mathcal{U}_{0}^{i}+D_{i}\mathcal{V}^{i},$%
\begin{equation}
\mathcal{V}^{i}\equiv\frac{1}{2}\mathcal{R}_{e}\{\tilde{\psi}_{0}[\gamma
^{j},\gamma^{i}]D_{j}\psi_{1}+\tilde{\psi}_{1}[\gamma^{j},\gamma^{i}]D_{j}%
\psi)\}
\end{equation}
Embedding and multiplication properties of Sobolev spaces give
\[
\tilde{\psi}_{1}[\gamma^{j},\gamma^{i}]D_{j}\psi)\in H_{s,\delta}%
\times\{C_{n-1}^{1}\cap H_{s-1,\delta+1}\}\subset H_{s-1,2\delta+1+\frac{n}%
{2}}\subset C_{\beta}^{0}%
\]%
\[
\beta<2\delta+1+\frac{n}{2}+\frac{n}{2}<2n-3.
\]
To estimate the other term one remarks that
\begin{equation}
D_{i}\{\tilde{\psi}_{0}[\gamma^{j},\gamma^{i}]D_{j}\psi_{1}\}\equiv D_{i}%
D_{j}\{\tilde{\psi}_{0}[\gamma^{j},\gamma^{i}]\psi_{1}\}-D_{i}\{D_{j}%
\tilde{\psi}_{0}[\gamma^{j},\gamma^{i}]\psi_{1}\}
\end{equation}
the first parenthesis is an antisymmetric 2-tensor hence its double divergence
$D_{i}D_{j}$ is identically zero. The second parenthesis is
\[
D_{j}\tilde{\psi}_{0}[\gamma^{j},\gamma^{i}]\psi_{1}\in C_{n-1}^{0}\times
H_{s,\delta}\subset C_{\beta}^{0}.
\]
The Stokes formula implies, with $M_{r}^{n}:=M^{n}-\{\Omega_{I}\cap\sum
(x^{i})^{2}\geq r^{2}\}$
\begin{equation}
\int_{M_{r}^{n}}D_{i}\mathcal{U}^{i}\mu_{g}=\int_{S_{r}^{n-1}}\mathcal{U}%
^{i}n_{i}\mu_{\bar{g}}=\int_{S_{r}^{n-1}}(\mathcal{U}_{0}^{i}+\mathcal{V}%
^{i})n_{i}\mu_{\bar{g}}.
\end{equation}
The fall off properties found for $\mathcal{V}^{i}$ complete the proof.
\end{proof}

\begin{lemma}
If $(M^{n},g)$ is A.E. $R\geq0$ and $\psi_{0}$ is a smooth spinor constant in
$\Omega_{I}$ and zero in the other ends there exists on $M^{n}$ a spinor
$\psi\equiv\psi_{0}+\psi_{1},$ such that\ $\mathcal{D}\psi=0,$ \ $\psi_{1}\in
H_{s,\delta}.$
\end{lemma}

\begin{proof}
The hypotheses made on $\psi_{0}$ show that $\mathcal{D}\psi_{0}\in
H_{s-1,\delta+1}.$ Theorem 1 implies the existence of $\psi_{1}.$
\end{proof}

The lemmas imply that $m\geq0$ if $R\geq0,$ that is if ($M^{n},g)$ is a
maximal submanifold of $\mathbf{(M}^{n+1}\mathbf{,g);}$ equivalently, if the
pointwise gravitational momentum $P$ on $M^{n}$ has a wanishing trace. We will
now lift this restriction, proving moreover that $m\geq|p|$.

\section{Positive energy.}

We define a real vector $\mathcal{P}$ on an A.E. $(M^{n},g,K)$
by\footnote{{\footnotesize Remark that we do not to introduce a matrix
}$\gamma_{0}$}
\[
\mathcal{P}^{i}:=\frac{1}{2}\tilde{\psi}\gamma_{h}P^{ih}\psi\equiv\frac{1}%
{2}\tilde{\psi}(\gamma_{h}K^{ih}-\gamma^{i}\gamma^{j}\gamma^{h}K_{jh}%
)\psi,\text{ \ }P^{ih}=K^{ih}-\delta^{ih}\mathrm{trK}%
\]
If $\psi_{0}$ is as before a smooth spinor constant in one end of $M^{n}$ and
zero in the other ends and $\psi$ is a spinor on $M^{n}$ such that $\psi
-\psi_{0}\in C_{\beta}^{0},$ $\beta>0,$ then
\begin{equation}
\lim_{r\rightarrow\infty}\int_{S_{r}^{n-1}}\mathcal{P}^{i}n_{i}\mu_{\bar{g}%
}=\frac{1}{2}\tilde{\psi}_{0}\gamma_{h}p^{h}\psi_{0},\text{ \ \ }p^{h}%
:=\lim_{r\rightarrow\infty}\int_{S_{r}^{n-1}}P^{ih}n_{i}\mu_{\bar{g}}.
\end{equation}
It is elementary to check using the properties of the $\gamma^{\prime}s$ that
$\gamma_{h}p^{h}$ is an hermitian matrix with eigenvalues $\pm|p|.$ If we
choose for $\psi_{0}$ an eigenvector of the eigenvalue $-|p|$ we then have
\begin{equation}
\tilde{\psi}_{0}\gamma_{h}p^{h}\psi_{0}=-|\psi_{0}|^{2}\text{\ }|p|.
\end{equation}

To estimate the limit 5.1. we use again the Stokes formula, with
\begin{equation}
D_{i}\mathcal{P}^{i}\equiv\frac{1}{2}D_{i}(\tilde{\psi}\gamma_{h}P^{ih}%
\psi)\equiv\frac{1}{2}D_{i}(\tilde{\psi}\gamma_{h}\psi\text{)}P^{ih}+\frac
{1}{2}\tilde{\psi}\gamma_{h}\psi D_{i}P^{ih}.
\end{equation}
The momentum constraint 0.2 gives
\[
D_{i}P^{ih}=-T_{0}^{h}.
\]
On the other hand, the identity 4.4 together with the Hamiltonian constraint
implies that
\begin{equation}
D_{i}\mathcal{U}^{i}\equiv|D\psi|^{2}-|\mathcal{D}\psi|^{2}+(\frac{1}{2}%
T_{00}+\frac{1}{4}|K|^{2}-\frac{1}{4}|\mathrm{tr}K|^{2})|\psi|^{2}.
\end{equation}
We introduce the notations\footnote{{\footnotesize Note that }$\not \nabla
${\footnotesize \ is a linear operator mapping space spinor into space
spinors, not the trace on }$M^{n}${\footnotesize \ of the covariant derivative
of a spacetime spinor.}}
\begin{equation}
\nabla_{i}\psi:=D_{i}\psi+\frac{1}{2}\gamma^{h}K_{ih}\psi,\text{
\ }\not \nabla\psi:=\gamma^{i}\nabla_{i}\equiv(\mathcal{D}+\frac{1}%
{2}\mathrm{tr}K)\psi.
\end{equation}
Elementary computation using $D_{i}\eta^{hj}\equiv0,$ $D_{i}\gamma^{h}\equiv0
$ gives
\begin{equation}
|\nabla\psi|^{2}:=\eta^{ij}\widetilde{\nabla_{i}\psi}\nabla_{j}\psi
\equiv|D\psi|^{2}+\frac{1}{2}D_{i}(\tilde{\psi}\gamma_{h}\psi)K^{ih}+\frac
{1}{4}|K|^{2}|\psi|^{2}%
\end{equation}
The identity 5.4 can therefore be written after simplification
\begin{equation}
D_{i}\mathcal{U}^{i}\equiv|\nabla\psi|^{2}-|\mathcal{D}\psi|^{2}+(\frac{1}%
{2}T_{00}-\frac{1}{4}|\mathrm{tr}K|^{2})|\psi|^{2}-\frac{1}{2}D_{i}%
(\tilde{\psi}\gamma_{h}\psi)K^{ih}).
\end{equation}
We deduce from the definition
\[
|\not \nabla\psi|^{2}\equiv|\mathcal{D}\psi|^{2}+\frac{1}{2}D_{i}%
(\widetilde{\psi}\gamma^{i}\psi)\mathrm{tr}K+\frac{1}{4}|\mathrm{tr}%
K|^{2}|\psi|^{2}%
\]
which gives
\begin{equation}
D_{i}\mathcal{U}^{i}\equiv|\nabla\psi|^{2}-|\not \nabla\psi|^{2}+\frac{1}%
{2}T_{00}|\psi|^{2}-\frac{1}{2}D_{i}(\tilde{\psi}\gamma^{h}\psi)P_{ih}.
\end{equation}

\begin{lemma}
If $(M^{n},g,K)$ is A.E. then it holds that
\begin{equation}
D_{i}(\mathcal{U}^{i}+\mathcal{P}^{i})\geq0
\end{equation}
if the dominant energy condition holds and $\psi$ satisfies the equation
$\not \nabla\psi=0.$
\end{lemma}

\begin{proof}
The identities 5.3\ and 5.8 lead to
\begin{equation}
D_{i}(\mathcal{U}^{i}+\mathcal{P}^{i})\equiv|\nabla\psi|^{2}-|\not \nabla
\psi|^{2}+\mathcal{T},\text{ \ }\mathcal{T}:=\frac{1}{2}(T_{00}|\psi
|^{2}-\tilde{\psi}\gamma^{h}\psi T_{0h}).\text{\ }%
\end{equation}
with $\mathcal{T}\geq0$ under the dominant energy condition, because
\[
|\tilde{\psi}\gamma^{h}\psi T_{0h}|\equiv|\psi|^{2}(\eta^{ih}T_{0i}%
T_{0h})^{\frac{1}{2}}\leq T_{00}|\psi|^{2}|
\]
with $\mathcal{T}\geq0$ under the dominant energy condition, because
\[
|\tilde{\psi}\gamma^{h}\psi T_{0h}|\equiv|\psi|^{2}(\eta^{ih}T_{0i}%
T_{0h})^{\frac{1}{2}}\leq T_{00}|\psi|^{2}|.
\]
\end{proof}

\begin{lemma}
If $(M^{n},g,K)$ is A.E., then the equation $\not \nabla\psi=0$ has a solution
$\psi\equiv\psi_{0}+\psi_{1},$ $\psi_{0}$ smooth, constant in $\Omega_{I}$ and
zero in the other ends, and $\psi_{1}\in H_{s,\delta}.$
\end{lemma}

\begin{proof}
The operator $\not \nabla$ has the same principal part as $\mathcal{D}$,
therefore is also elliptic. It maps $H_{s,\delta}$ into $H_{s-1,\delta+1}.$
The equation $\not \nabla\psi_{1}=-\not \nabla\psi_{0}\in H_{s-1,\delta+1}$
has one and only one solution if $\not \nabla$ is injective on $H_{s,\delta}.$
To show injectivity\footnote{{\footnotesize See a simlar proof in Chrusciel
Krakow lecture notes.}} we remark that the identity 5.10 was established
without restriction on $\psi,$ starting from the definitions of \ $\mathcal{U}%
^{i}$ and $\mathcal{P}^{i}.$\ We make $\psi=\psi_{1}$ in 5.10 and integrate it
on $M^{n},$ the fall off of $\psi_{1}$ implies that the divergence gives no
contribution, the equation $|\not \nabla\psi_{1}|^{2}=0$ implies therefore
that on $M^{n},$ if $\mathcal{T}$\ $\geq0$%
\[
|\nabla\psi_{1}|^{2}=0,\text{\ \ i.e. \ \ }D_{i}\psi_{1}+\frac{1}{2}\gamma
^{h}K_{ih}\psi_{1}=0,\text{ \ with \ }\gamma^{h}K_{ih}\in H_{s-1,\delta+1}%
\]
The Poincar\'{e} inequality\footnote{{\footnotesize See for instance
Y.\ Choquet-Bruhat ''General Relativity and the Einstein equations'' appendix
3sobolev spaces'' p.541}} in weighted \ Hilbert spaces leads to $\psi_{1}=0$
if $\psi_{1}\in H_{s,\delta},$ $s>\frac{n}{2}+1$ and $-2+\frac{n}{2}%
>\delta>-\frac{n}{2}.$
\end{proof}

The lemmas, after choice of $\psi_{0}$ satisfying 5.2, prove the following theorem

\begin{theorem}
If an Einsteinian spacetime satisfies the dominant energy condition, the
energy momentum vector $E^{0}=m,$ $E^{i}=p^{i}$ of each end of an A.E. slice
($M^{n},g,K)$ satisfies the inequality
\[
m\geq|p|
\]
\end{theorem}

\begin{acknowledgement}
I thank Thibault Damour for his encouragement to write this Note and for his
help in doing it. I thank Piotr Chrusciel for communicating to me a text of
his 2010 lectures in Krakow. A long list of references can be found there.
\end{acknowledgement}

16 Avenue d'Alembert, 92160, Antony

01 47 02 50 86, ycb@ihes.fr

Rubrique; Physique Math\'{e}matique
\end{document}